\documentclass[11pt]{amsart}

\usepackage[latin1]{inputenc}
\usepackage[T1]{fontenc}
\usepackage{amsmath,amsfonts,amssymb,amsthm,url,geometry,mathrsfs}

\newcommand{\e}{\varepsilon}
\renewcommand{\phi}{\varphi}
\newcommand{\iy}{\infty}

\DeclareMathOperator{\tr}{Tr}

\DeclareMathOperator{\osc}{osc}

\renewcommand{\leq}{\leqslant}
\renewcommand{\geq}{\geqslant}

\newcommand{\cB}{\mathcal{B}}
\newcommand{\cD}{\mathcal{D}}

\newcommand{\cH}{\mathcal{H}}

\newcommand{\cW}{\mathcal{W}}

\newcommand{\st}{\textnormal{ s.t }}
\newcommand{\N}{\mathbf{N}}
\newcommand{\Z}{\mathbf{Z}}

\newcommand{\R}{\mathbf{R}}
\newcommand{\C}{\mathbf{C}}
\DeclareMathOperator{\E}{\mathbf{E}}
\DeclareMathOperator{\card}{card}
\DeclareMathOperator{\spec}{spec}

\renewcommand{\P}{\mathbf{P}}
\newcommand{\M}{\mathcal{M}}

\newcommand{\mN}{\mathscr{N}}
\newcommand{\mU}{\mathcal{U}}
\newcommand{\Id}{\mathrm{Id}}

\newcommand{\ketbra}[2]{| #1 \rangle \langle #2 |}
\newcommand{\ket}[1]{| #1 \rangle}
\newcommand{\bra}[1]{\langle #1 |}

\theoremstyle{plain}
\newtheorem{thm}{Theorem}[] 
\newtheorem{cor}[thm]{Corollary}
\newtheorem{prop}[thm]{Proposition}
\newtheorem{lem}[thm]{Lemma}

\newtheorem{fact}[thm]{Fact}

\theoremstyle{definition}

\theoremstyle{remark}

\begin{document}

\title{Hastings's additivity counterexample via Dvoretzky's theorem}
\author{Guillaume Aubrun} 
\address{Institut Camille Jordan, Universit\'e Claude Bernard Lyon 1, 43 boulevard du 11 novembre 1918, 69622 Villeurbanne CEDEX, France}
\email{aubrun@math.univ-lyon1.fr}

\author{Stanis\l aw Szarek}
\address{Equipe d'Analyse Fonctionnelle, 
Institut de Math\'ematiques de Jussieu,              
Universit\'e Pierre et Marie Curie-Paris 6,          
4 place Jussieu
75252 Paris, France  {\sl and} 
Department of Mathematics, Case Western Reserve University, Cleveland, Ohio 44106, USA 
}
\email{szarek@math.jussieu.fr}

\author{Elisabeth Werner}
\address{Department of Mathematics, Case Western Reserve University, Cleveland, Ohio 44106, USA \ {\sl and} Universit\'{e} de Lille 1, UFR de Math\'{e}matique, 59655 Villeneuve d'Ascq, France
}
\email{elisabeth.werner@case.edu}
\thanks{The research of the first named author was partially supported by  the {\itshape Agence Nationale de la Recherche} grant ANR-08-BLAN-0311-03. The research of the second and third named authors was partially supported by their respective grants from the {\itshape National Science Foundation} (U.S.A.) and from the 
{\itshape U.S.-Israel Binational Science Foundation}. 
The authors would like to thank M. B. Hastings and M. Horodecki for valuable comments, 
and {\itshape MF Oberwolfach} -- where insights crucial to this project were crystallized 
-- for their hospitality.
}

\begin{abstract}
The goal of this note is to show that Hastings's counterexample to the additivity of minimal
output von Neumann entropy can be 
readily deduced from a sharp version of Dvoretzky's theorem.
\end{abstract}

\maketitle

\section*{Introduction}
A fundamental problem in Quantum Information Theory is to determine the 
capacity of a quantum channel to transmit classical information. 
The seminal Holevo--Schumacher--Westmoreland theorem 
expresses this capacity as a regularization of the so-called Holevo $\chi$-quantity 
(which gives the one-shot 
%single-use
capacity) over multiple uses of the channel; see, e.g., \cite{nc}. 
This extra step 
could have been skipped 
%would not have been necessary 
 if the $\chi$-quantity had been additive, i.e., if
\begin{equation}
 \label{additivity-chi}
\chi(\Phi \otimes \Psi) = \chi(\Phi) + \chi(\Psi) 
\end{equation}
for every pair $(\Phi,\Psi)$ of quantum channels. 
It would have then followed that the $\chi$-quantity  and the capacity coincide, yielding a 
single-letter formula for the latter. 
Determining the veracity of \eqref{additivity-chi} had been a major open problem
for at least a decade (we refer, e.g., to the survey \cite{holevo-icm}). 
A substantial progress was made by Shor \cite{shor} who showed that 
\eqref{additivity-chi} was formally equivalent to the additivity
of the minimal output von Neumann entropy of quantum channels --- 
a much more tractable quantity. Using this equivalence, the equality  
\eqref{additivity-chi} was eventually shown to be false by Hastings 
\cite{hastings}, with appropriate randomly constructed channels as 
a counterexample. 

In this note, we %would like to 
revisit Hastings's counterexample from the 
viewpoint of Asymptotic Geometric Analysis (AGA). This field --- 
originally an offspring of Functional Analysis --- aims at studying 
geometric properties of convex bodies 
(or equivalently, norms) in spaces of high (but finite) dimension. 
More specifically, our goal is to show that (a variant of) Hastings's 
analysis can be rephrased in the language of AGA, and his result 
 deduced with only minor effort from a sharp version
of Dvoretzky's theorem \cite{dvoretzky} on almost spherical sections of convex bodies 
--- a fundamental result of AGA. 
This makes the argument much more transparent and will hopefully lead 
to a better understanding of the problem of capacity.
Our approach is largely inspired by Brandao--Horodecki \cite{bh}, who were able 
to reformulate Hastings's analysis in the framework of concentration of measure.

\section*{Notation}

Throughout the paper, the letters $C$, $c$, $C'$, ... denote absolute 
positive constants, independent of the instance of the problem 
(most notably of the dimensions involved),  whose values may 
change from occurrence to occurrence. 
The values of these constants can be computed by reverse-engineering 
the argument, but we will not pursue this task. 
We also use the following convention: whenever a formula is given 
for the dimension of a (sub)space, it is  tacitly understood that one 
should take the integer part.

Let  $\M_{k,d}$ be the space of $k\times d$ matrices (with complex entries), and $\M_d=\M_{d,d}$. 
More generally, $\mathcal{B}(\cH)$ will stand for the space of (bounded) linear operators on the Hilbert space $\cH$.  
We will write $\|\cdot\|_{p}$ for the Schatten $p$-norm $\| A \|_p = \big( \tr (A^\dagger A)^{p/2} \big)^{1/p}$. The limit case $\|\cdot \|_\iy$ is the operator (or ``spectral'') norm, while 
 $\|\cdot\|_{HS}=\|\cdot\|_{2}$ is the Hilbert--Schmidt (or Frobenius) norm.
Let $\cD(\C^d)$ be the set of {\em density matrices} on $\C^d$, i.e., positive semi-definite trace one operators on $\C^d$ (or {\em states} on $\C^d$).
%, or -- more properly -- states on $\cB(\cH)$). 
If $\rho$ is a state on $\C^d$, its {\em von Neumann entropy} $S(\rho)$ is defined as $S(\rho)=-\tr
\rho \log \rho$. 
If $\Phi : \M_m \to \M_{k}$ is a {\em quantum channel} (completely
positive trace preserving map), its {\em minimal output entropy} is
\[ S_{\min}(\Phi) = \min_{\rho \in \cD(\C^m)} S(\Phi(\rho)). \]
Concavity of $S$ implies that the minimum is achieved on a pure state.

\section*{Channels as subspaces}

The crucial insight allowing to relate analysis of quantum channels to
high-dimensional convex geometry is the observation that there is an
essentially one-to-one correspondence between channels and linear subspaces
of composite Hilbert spaces. Specifically, let  $\cW$ be a subspace of $\C^k \otimes \C^d$ of dimension $m$. 
Then $\Phi : \cB(\cW) \to \M_k$ defined by $\Phi(\rho) = \tr_{\C^d}(\rho)$  is 
a  quantum channel; here $\tr_{\C^d}$ is the {\em partial trace} with respect to the second factor in $\C^k \otimes \C^d$.  Alternatively, and perhaps more properly, we could identify 
$\cW$ with $\C^m$ via an isometry $V: \C^m \to \C^k \otimes \C^d$ whose range is $\cW$ 
and define, 
for $\rho \in \M_m$,  the corresponding channel $\Phi: \M_m \rightarrow  \M_k$ by 
\begin{equation}\label{Phi}
\Phi(\rho) = \tr_{\C^d}(V\rho V^\dagger).
\end{equation}

It is now easy to  define a natural family of  random quantum channels. 
%For appropriately chosen values for the parameters 
They will be associated, via the above scheme, to random $m$-dimensional 
subspaces  $\cW$ of $\C^k \otimes \C^d$, 
distributed according to the Haar measure on the corresponding
Grassmann manifold 
%of $m$-dimensional subspaces of $\C^k \otimes \C^d$ 
(for some fixed positive integers $m,d,k$ that will be specified later).  
Note that all reasonable parameters of a channel defined by (\ref{Phi}) such as $ S_{\min}(\Phi) $ 
depend only on the subspace $\cW = V(\C^m)$ and not on a particular choice 
of the isometry $V$ (this will be also obvious from what follows). In particular, 
the language of ``random $m$-dimensional subspaces  of $\C^k \otimes \C^d$\,''
is equivalent to that of ``random isometries from $\C^m$ to $\C^k \otimes \C^d$.''

\section*{The additivity conjectures and the main theorem}

The following question has attracted 
considerable attention in the last few years: 
if $\Phi$ and $\Psi$
are two quantum channels, is it true that 
\begin{equation} \label{additivity-smin}
S_{\min}(\Phi \otimes \Psi) = S_{\min}(\Phi)
+ S_{\min}(\Psi) \ ?
\end{equation}
Shor \cite{shor} showed it to be formally equivalent to a number of central 
questions in quantum information theory, including
the additivity of the $\chi$-quantity mentioned in the introduction.

Note that the inequality ``$\leq$'' always holds (consider product input states). 
However, as was first proved by 
Hastings using random constructions \cite{hastings}, 
the reverse inequality is false
 in general. 
The exegesis of Hastings's argument has subsequently been carried out in \cite{bh} and \cite{FKM}. 
We will show here that the analysis of (a variant of) Hastings's example essentially amounts to
applying the right version of Dvoretzky's theorem 
and leads to the conclusion that 
high-dimensional random channels typically violate \eqref{additivity-smin}.

\begin{thm} \label{maintheorem}
Let $k \in \N$, $m=ck^2$ and $d=Ck^2$ ($c$ and $C$ being appropriate absolute constants).
 Let $V : \C^m \to \C^k \otimes \C^d$ be a random isometry and $\Phi: \M_m \to \M_k$ be the corresponding random channel given by {\rm (\ref{Phi})}.
Then for $k$ large enough, with large probability,
\[ S_{\min} ( \Phi \otimes \bar{\Phi} )  < S_{\min} (\Phi) + S_{\min} (\bar{\Phi}).
\]
\end{thm}
The expression ``with large probability'' in Theorem \ref{maintheorem} and in what follows 
may be understood as ``with probability  $> \theta$, where 
$\theta \in  (0,1)$ is arbitrary but fixed in advance''  
 (note that, in particular, the threshold value of $k$ could then depend on $\theta$). 
 However, much stronger assertions are in fact true, for example the probability of the 
 exceptional set in Theorem \ref{maintheorem} can be majorized by $\exp(-c'm)$. 
Another comment: 
one only uses in the proof that $m$ and $d$ are comparable, and larger than $ck^2$.
%essential to the proof are $ck^2\leq m \leq C d \leq ckm$. 

The proof will be based on separately majorizing $S_{\min} ( \Phi \otimes \bar{\Phi} )$, 
%the minimal output entropy of the composite channel, 
which is done via a well-known and relatively simple trick, 
and on minorizing $S_{\min} (\Phi) = S_{\min} (\bar{\Phi})$, which is the main point of the argument.  

\smallskip

A question analogous  to \eqref{additivity-smin} can be asked for the minimal 
output $p$-R\'enyi entropy $(p>1)$. For the
additivity of R\'enyi entropy, random counterexamples were constructed earlier 
by Hayden--Winter \cite{hw}. It was shown in \cite{asw} that the Hayden--Winter 
analysis can also be simplified (at least conceptually) by appealing to Dvoretzky's theorem.
Working with the von Neumann entropy, however, requires more effort. First, while 
\cite{asw} relied on a straightforward instance of Milman's ``tangible'' version \cite{milman, FLM} of 
Dvoretzky's theorem 
for Schatten classes that was documented in the literature already  in the 1970's, 
we now need a more subtle,  sharp version (which appears in the literature only 
implicitly). Second, this sharp version is not applied in the most direct way 
and requires additional preparatory work (for which we 
mostly follow the approach of Brandao--Horodecki~\cite{bh}).

\section*{Lower bound for  $S_{\min} (\Phi)$ : the approach}

Since we are going to consider channels with near-maximal minimal output entropy, the following simple inequality
(Lemma III.1 in \cite{bh}, or formula (40) in \cite{hastings}) will allow to replace 
the analysis of the von Neumann entropy $S$ by that of a smoother quantity.

\begin{lem}
For every state $\sigma \in \cD(\C^k)$,
\vskip-2mm
\[ S(\sigma) \geq S\left(\frac{\Id}k\right) -k \left\|\sigma-\frac{\Id}k\right\|_{HS}^2  . \] 
Consequently, for every quantum channel $\Phi : \M_m \to \M_k$,
\begin{equation} \label{Smin-HSmax} S_{\min}(\Phi) \geq \log (k) -k \cdot \max_{\rho \in \cD(\C^m)} \left\|\Phi(\rho)-\frac{\Id}k\right\|_{HS}^2. \end{equation}
\end{lem}
%
%
%\section*{Tensors as matrices}

It will be convenient to identify $\C^k \otimes \C^d$ (or, to be more precise, 
$\C^k \otimes \overline{\C^d}$ --- a distinction we will ignore) with $\M_{k,d}$ via the canonical map induced by $u\otimes v \to \ketbra{u}{v}$. If $x \in \C^k \otimes \C^d$ is so identified with a matrix $M \in \M_{k,d}$, then 
\begin{equation}\label{PhiMatrix}
\tr_{\C^d} \ketbra{x}{x} = MM^\dagger .
\end{equation}
Via this identification, Schmidt coefficients of $\ket{x}$ coincide with singular values of $M$. While the tensor
and matrix formalisms are equivalent,  the matrix formalism is arguably more transparent, which 
sometimes leads to simpler arguments.  

Denote by $\cW \subset \C^k \otimes \C^d$ the subspace inducing $\Phi$.  
Note that the maximum in   \eqref{Smin-HSmax} is necessarily attained 
on pure states which, in this identification, correspond to unit vectors  $x \in \cW$. 
For such states the action of $\Phi$ is given --- in the matrix formalism --- by  \eqref{PhiMatrix}, 
and so  the inequality \eqref{Smin-HSmax} can be rewritten as
\begin{equation} S_{\min}(\Phi) \geq \log(k) -k \cdot \max_{M \in \cW, \, \|M\|_{HS}=1} \left\|MM^\dagger-\frac{\Id}k\right\|_{HS}^2.
\label{Smin-HSmax-matrix} 
\end{equation}
The idea will be to show that, for a random subspace $\cW$, the maximum on the right is very small; 
this will be formalized in the next proposition.

\section*{The main proposition and the derivation of the main theorem}

The heart of the argument is the following proposition

\begin{prop} \label{mainprop}
There are absolute constants $c,C,C'>0$ so that for every 
$k$, for $d = Ck^2$  and $m = cd$, 
a random Haar-distributed subspace $\cW$ of dimension $m$ in $\M_{k,d}$ satisfies
\begin{equation}\label{bound}
\max_{M \in \cW, \|M\|_{HS}=1} \left\|MM^\dagger-\frac{\Id}{k} \right\|_{HS} \leq \frac{C'}{k} 
\end{equation}
with large probability (tending to $1$ when $k$ tends to $\iy$).
\end{prop}

From the proposition one quickly deduces that the pair $(\Phi, \bar{\Phi})$ is a 
counterexample to the additivity of minimum output
 von Neumann entropy. Indeed, a straightforward calculation shows that 
 applying $\Phi \otimes \bar{\Phi}$  to the maximally entangled state
yields an output state with one eigenvalue greater than or equal to $\frac{\dim \cW}{\dim \M_{k,d}}=\frac m{kd}=\frac c k$ (\cite{hw}, Lemma III.3; see also section 6 in \cite{cn}). 
Then, a simple argument using just concavity of $S(\cdot)$ reduces the problem to calculating 
the entropy of the state with one eigenvalue {\sl equal } to $\frac c k$ and all the remaining ones 
identical, which yields
\[ S_{\min}(\Phi \otimes \bar{\Phi}) \leq 2 \log{k} - \frac{c \log{k}}{k} +\frac{1}{k}
\]
On the other hand, equation \eqref{Smin-HSmax-matrix} together with  Proposition \ref{mainprop} implies
\[ S_{\min}(\Phi) \geq \log(k)-\frac{C'^2}{k}. \]
Since $S_{\min}(\bar{\Phi})=S_{\min}(\Phi)$, the inequality of Theorem \ref{maintheorem}  
follows if $k$ is large enough, as required.

\section*{Dvoretzky's theorem : take one}

We wish to point out that while Proposition \ref{mainprop} will be {\em derived from} a 
 Dvoretzky-like theorem for Lipschitz functions  (Theorem \ref{dvoretzky-sharp} below), 
 it can be {\em rephrased} in the language of the standard Dvoretzky's theorem. 
 Indeed,  its assertion says that for every $M \in \cW$ with $ \|M\|_{HS}=1$ we have 
\begin{equation} \label{4thmoment}
\frac{C^2}{k^2} \geq \left\|MM^\dagger-\frac{\Id}{k} \right\|_{HS}^2 = 
\tr |M|^4 - \frac{2 \tr MM^\dagger}k +\frac {\tr{\Id}}{k^2} =\tr |M|^4 -\frac 1k \geq 0 .
\end{equation}
Consequently,
\begin{equation} 
k^{-1/4}\|M\|_{HS} \leq \|M\|_4 \leq k^{-1/4} \Big( 1 + \frac {C^2}k\Big)^{1/4}\|M\|_{HS} 
\leq k^{-1/4} \Big( 1 + \frac {C^2}{4k}\Big)\|M\|_{HS} 
\label{dvor1} 
\end{equation}
%(where $\|\cdot \|_p$ is the Schatten $p$-norm) 
for all $M \in \cW$.  In other words, 
$\cW$ is $(1 +\delta)$-Euclidean, with $\delta = \frac {C^2}{4k}$, when considered 
as a subspace  of the normed space $\big(\M_{k,d}, \|\cdot\|_4\big)$, the 
Schatten $4$-class. 

In our prior work \cite{asw} we similarly observed that the crucial technical step of the 
Hayden-Winter proof of non-additivity of $p$-R\'enyi entropy for $p>1$ can be 
restated as an instance of Dvoretzky's
theorem for the Schatten $2p$-class. 
There is an important difference, however.  While in the case of $p$-R\'enyi entropy 
the needed Dvoretzky-type statement was known since the 1970s, for the statement 
of the type \eqref{dvor1} needed in the present context, the ``off the shelf'' methods seem to yield only 
$\delta = O(k^{-1/4})$ as opposed to $\delta = O(k^{-1})$  above.  
This also suggests that while for the $p$-R\'enyi entropy derandomization of the 
example --- i.e., supplying {\sl explicit } channels for which the additivity fails ---  
may be a feasible project (see section IX in \cite{asw} and references therein), the analogous 
task for the von Neumann  entropy is~likely~to~be~much~harder.

\section*{Dvoretzky's theorem : take two}

We use the following definitions: if $f$ is a function from a metric space $(X,d)$ to $\R$, and $\mu \in \R$, 
the {\itshape oscillation} of $f$ around $\mu$ on a subset $A \subset X$ is
\[ \osc(f,A,\mu) = \sup_{A} |f-\mu|. \]
A function $f$ defined on the unit sphere $S_{\C^n}$  is called {\itshape circled} 
if $f(e^{i\theta}x)=f(x)$ for any $x \in S_{\C^n}, \theta \in [0,2\pi]$. 
If $X$ is a real random variable, we will say that $\mu$ is a {\em central value} of $X$ if 
$\mu$ is either the mean of $X$, or any number between the 1st and the 3rd quartile of $X$
(i.e., if $\min \{\P(X\geq\mu), \P(X\leq\mu)\} \geq \frac 14$; 
this happens in particular if $\mu$ is the median of $X$).
\vskip 3mm
We will need the following variant of Milman's ``tangible'' version of~Dvoretzky's~theorem.

\begin{thm}[Dvoretzky's theorem for Lipschitz functions]
\label{dvoretzky-sharp}
If $f : S_{\C^n} \to \R$ is a 1-Lipschitz circled function, then for every $\e >0$, 
if $E \subset \C^n$ is a random subspace (Haar-distributed) of dimension $c_0 n \e^2$, we have with large probability
\[ \osc(f,S_{\C^n} \cap E,\mu) \leq \e ,\]
where $\mu$ is any central value of $f$  
(with respect to the normalized Lebesgue measure on $S_{\C^n}$) 
and $c_0$ is an absolute constant.
If the function is  $L$-Lipschitz, the dimension changes to $c_0 n (\e/L)^2$.
\end{thm}

A striking application of the theorem above is to the case when $f$ is the gauge function 
of a convex body, or a norm: it leads to the fact that any high-dimensional convex 
body has almost spherical sections.

At the heart of Dvoretzky-like phenomena lies the concentration of measure, which in our 
framework is expressed by

\begin{lem}[L\'evy's lemma \cite{levy}] \label{levy}
If $f:S^{n-1} \to \R$ is a $1$-Lipschitz function, then for every $\e>0$,
\[ \P ( |f(x)-\mu| > \e ) \leq C_1 \exp(-c_1n \e^2) ,\]
where $x$ is uniformly distributed on $S^{n-1}$,  $\mu$ is 
any central value of $f$, 
and $C_1,c_1>0$ are absolute constants.
\end{lem}

Results such as Theorem \ref{dvoretzky-sharp} or  L\'evy's lemma are usually 
stated with $\mu$ equal to the median or the mean of $f$. 
However, once we know that the result is true for {\em some} central value 
(or, for that matter, for {\em any} $\mu \in \R$), 
it holds {\em a posteriori } 
for {\em any}  such value (up to changes in the constants)  as,  for 
$1$-Lipschitz functions,  all central values differ at most by 
$C/\sqrt{n}$.

The obvious idea to prove Theorem \ref{dvoretzky-sharp} is to use  L\'evy's lemma and an $\e$-net argument ---
 using the fact that an $\e$-net in $S_{\C^n} = S^{2n-1}$ can be chosen to have cardinality $\leq (1+2/\e)^{2n}$ (see \cite{pisier},
Lemma 4.10). Indeed, this was essentially Milman's original argument in \cite{milman}. 
However,  one only obtains this way a 
subspace $E$ of dimension $c n \e^2/\log(1/\e)$. 
For many applications (including our previous paper \cite{asw}), this extra
logarithmic factor is not an issue. 
However, in the present case,  having the optimal dependence on $\e$ is crucial. 

\smallskip
The classical framework of convex geometry is the real case (with or without the assumption ``circled,'' which in that context just means then that the function is even). In that setting, Theorem 
\ref{dvoretzky-sharp} was proved by Gordon \cite{gordon} who used 
comparison inequalities for Gaussian processes. A proof based on concentration
of measure was later given by Schechtman \cite{schechtman}. The complex case does not seem to appear in the literature. Actually, at the face of it, 
Gordon's proof does not extend to the complex setting, while Schechtman's proof does. We sketch Schechtman's proof 
of Theorem \ref{dvoretzky-sharp} in Appendix A. It is not clear whether the assumption ``$f$ circled'' in Theorem \ref{dvoretzky-sharp} can be completely removed; we do know that it is needed at most for very small values of $\e$.

\section*{Proof of the main proposition}

Let $S_{HS}$ be the Hilbert--Schmidt sphere in $\M_{k,d}$ and let $M$ be a random matrix uniformly distributed on $S_{HS}$. Let $\tilde{g}(\cdot)$ be the function defined on $S_{HS}$ by
\[ \tilde{g}(M) = \left\| MM^\dagger - \frac{\Id}{k} \right\|_{HS}.\]

The next well-known lemma asserts that the singular values of a very rectangular random matrix are very concentrated. This is a familiar phenomenon in random matrix theory that goes back to 
\cite{MP}.  Versions of this lemma
 appeared in the QIT literature under the tensor formalism (see for example Lemma III.4 in \cite{hlw}). However, these versions
 typically introduce an unnecessary logarithmic factor which would imply that the main proposition holds with $d=Ck^2 \log k$
 instead of $d=Ck^2$.  For completeness, we include a proof of Lemma 
 \ref{concentration-spectrum} in Appendix B.

\begin{lem}
\label{concentration-spectrum}
There exist absolute constants $C,c>0$ such that, 
if $M$ is uniformly distributed on the Hilbert--Schmidt sphere in $\M_{k,d}$ ($d \geq C^2k$), then with probability larger than $1-\exp(-ck)$, 
\begin{equation} \label{spectrumbounded} \spec (MM^\dagger) \subset \left[ \left(\frac{1}{\sqrt{k}} - \frac{C}{\sqrt{d}}\right)^2, 
\left(\frac{1}{\sqrt{k}} + \frac{C}{\sqrt{d}}\right)^2 \right] .\end{equation}
\end{lem}

We note that inclusion \eqref{spectrumbounded} can be reformulated as follows: all singular values of $M$ differ from $1/\sqrt{k}$ by less than $C/\sqrt{d}$. (Recall that the singular values of $M$ correspond to the Schmidt coefficients
of a random pure state in $\C^k \otimes \C^d$.)

\smallskip
We will use in the sequel the following immediate corollary of Lemma 
 \ref{concentration-spectrum}.
 
 \begin{cor}
\label{norm4thmoment}
Under the hypotheses of Lemma  \ref{concentration-spectrum} and denoting $C_0=3C$\\
{\rm (a)}\ with probability larger than $1-\exp(-ck)$, all eigenvalues of 
$MM^\dagger$ differ from $1/k$ by less than $C_0/\sqrt{kd}$; 
consequently, the median (or any fixed quantile) of $\tilde{g}$ is bounded by $C_0/\sqrt{d}$ for $k$ large enough.\\
{\rm (b)}\ if $d \geq C^2k$, the median (or any fixed quantile) of $\|M\|_{\iy}$ is bounded by $2/\sqrt{k}$ for $k$ large
enough.
\end{cor}
\vskip-1mm
We point out that while we chose to present statements (a) and (b) above as consequences 
of Lemma  \ref{concentration-spectrum} for clarity and for ``cultural'' reasons 
(the lemma being familiar to the QIT community), more precise versions of 
these statements are available in (or can be readily deduced from) 
the random matrix literature. Re (a), the study of the distribution of $\tilde{g}$ 
is, by  \eqref{4thmoment}, equivalent to that of  $\tr |M|^4$, and a closed 
formula for the expected value of the latter is known (up to terms of smaller order, its value is   $1/k+1/d$); see, e.g., \cite{HT} (section 8) and its references. 
Re (b), sharp estimates on the tail of $\|M\|_{\iy}$ can also be found in  \cite{HT} (proof of Lemma 7.3), 
in particular every fixed quantile is $1/\sqrt{k}+1/\sqrt{d}$ up to terms of smaller order. 
This result can also be retrieved via methods of earlier papers \cite{geman, silver}, 
which focused on the real case. 

\smallskip
The function $\tilde{g}$ is $2$-Lipschitz on $S_{HS}$, and Corollary \ref{norm4thmoment}(a)  implies that the median of $\tilde{g}$ is
as small as we want for large $d$. However,  a direct application of Theorem \ref{dvoretzky-sharp}  
yields  only a bound of order $1/\sqrt{k}$ in \eqref{bound}. 
The trick --- already present
in the previous approaches --- is to exploit the fact that
$\tilde{g}$ has a much smaller Lipschitz constant when restricted to a certain large subset of $S_{HS}$.
As we will see, this bootstrapping argument is equivalent to applying Theorem  \ref{dvoretzky-sharp} {\sl twice}.

The following lemma appears in \cite{bh} with a rather long proof, but using the matrix formalism 
completely demystifies it.

\begin{lem}
\label{lipschitz-restriction}
The function $\tilde{g}$ is $6/\sqrt{k}$-Lipschitz when restricted to the set
\[ \Omega = \{ M \in S_{HS} \st \|M\|_{\iy} \leq 3/\sqrt{k}\} .\]
\end{lem}

\begin{proof}
The lemma is a consequence of the following chain of matrix inequalities
\begin{eqnarray*} 
\left\| MM^\dagger -  \frac{\Id}{k} \right\|_{HS} - \left\| NN^\dagger -   \frac{\Id}{k} \right\|_{HS} &\leq & 
\|MM^\dagger-NN^\dagger\|_{HS} \\
& \leq & \| M(M^\dagger-N^\dagger) + (M-N)N^\dagger \|_{HS} \\
& \leq & \| M \|_{\iy} \| M^\dagger - N^\dagger \|_{HS} + \|M-N\|_{HS} \|N^{\dagger}\|_{\iy} \\
& \leq & (\|M\|_{\iy} + \|N\|_{\iy}) \|M-N\|_{HS} 
\end{eqnarray*}
\end{proof}

The function $\|\cdot\|_{\iy}$ is $1$-Lipschitz on $S_{HS}$. By Corollary \ref{norm4thmoment}(b), its median is bounded by $2/\sqrt{k}$ for $d \geq C^2k$. 
(Note that L\'evy's lemma shows that the measure of the complement of $\Omega$ is very small.) 
An application of the standard Dvoretzky's theorem (i.e.,Theorem \ref{dvoretzky-sharp} for norms) to 
$f=\|\cdot\|_{\iy}$ with  $\mu$ equal to the median of $\|\cdot\|_{\iy}$
and with $\e=1/\sqrt{k}$ (note that the dimension of the ambient space is $n=kd$) 
shows that the intersection of $S_{HS}$ with a random subspace of dimension $cd$ in $\M_{k,d}$ is contained in $\Omega$ with large probability.

Let $g$ be a $6k^{-1/2}$-Lipschitz extension of $\tilde{g}_{|\Omega}$ to $S_{HS}$ --- in any metric space $X$, it is possible to extend any $L$-Lipschitz function $\tilde{h}$ defined on a subset $Y$ without increasing the Lipschitz constant; use, e.g., the formula
\[ h(x) = \inf_{y \in Y} \left[ \tilde{h}(y) + L\, {\rm dist} (x,y) \right].\]
This formula also guarantees that the extended function $g$ is circled. Since $g=\tilde{g}$ on most of $S_{HS}$, the median of $g$ (resp., $\tilde{g}$) is a central value of $\tilde{g}$ (resp., $g$).
We apply Theorem \ref{dvoretzky-sharp} to $g$ with $\e = 1/k$ and $L=6k^{-1/2}$ to get ($\mu$ being the median of $\tilde{g}$)
\[ \osc(g,S_{HS} \cap E,\mu) \leq 1/k . \]
on a random subspace 
$E \subset \M_{k,d}$ 
of dimension $m=c_0\cdot kd \cdot (k^{-1}/(6k^{-1/2}))^2=cd$. Using Corollary \ref{norm4thmoment}(a), we obtain that $\mu \leq 1/k$ for $d \geq (C_0 k)^2$. We then have
\[ \osc(g,S_{HS} \cap E,0) \leq 2/k . \]

If $S_{HS} \cap E \subset \Omega$ (which, as noticed before, holds with large probability), $g$ and $\tilde{g}$ coincide on $S_{HS} \cap E$ and therefore $\osc(\tilde{g},S_{HS} \cap E,0) \leq 2/k$. This completes the proof of Proposition \ref{mainprop} and hence that of Theorem \ref{maintheorem}.

\vskip.5cm

\section*{Appendix A : Proof of Theorem \ref{dvoretzky-sharp} (apr\`es  Schechtman)}

We sketch here a proof of Theorem \ref{dvoretzky-sharp}, essentially following Schechtman \cite{schechtman}. As we already mentioned, a simple use of a $\e$-net argument gives a parasitic factor $\log(1/\e)$. 
This can be improved by a {\itshape chaining} argument, which goes back (at least) to Kolmogorov  --- 
a way to use $\eta$-nets for all values of $\eta$ simultaneously. 

Consider the canonical inclusion $\C^m \subset \C^n$, and let $U \in \mU(n)$ be a random Haar-distributed unitary matrix. Then 
$F:=U(\C^m)$ is distributed according to the Haar measure on the Grassmann manifold of $m$-dimensional subspaces. If $f : S_{\C^n} \to \R$ is a $1$-Lipschitz circled function with mean $\mu$, we need to show that $\osc(f\circ U,S_{\C^m},\mu) \leq \e$ with large probability
provided $m \leq c_0 n \e^2$. We first prove a lemma.

\begin{lem} \label{lemma-subgaussian}
Let $f : S_{\C^n} \to \R$ be a $1$-Lipschitz circled function and $U \in \mU(n)$ be a Haar-distributed
random unitary matrix.
Then for any $x,y \in S_{\C^n}$ with $x \neq y$ and for any $\lambda>0$,
\[ \P( |f(Ux)-f(Uy)| > \lambda) \leq C \exp\left(-cn\frac{\lambda^2}{|x-y|^2}\right) \]
\end{lem}

\begin{proof}
Fix $x,y \in S_{\C^n}$. 
Since $f$ is circled (and $U$ is $\C$-linear), we may replace $y$ by $e^{i\theta}y$ and choose $\theta$ so that $\langle x | y \rangle$ is real nonnegative; note that this choice of $\theta$ minimizes $|x-y|$
and assures that $x+y$ and $y-x$ are orthogonal.  (This is the only {\em really} new point needed to acommodate the complex setting.) 
Set  $z=\frac{x+y}{2}$ and  $w=\frac{y-x}{2}$, 
then  $x=z+w$ and $y=z-w$. % and %$z'=\frac{z}{|z|}$, 
Further, set $\beta = |w|=\frac 12 |x-y|$  
(we may assume that $\beta\neq 0$) and $w'=\beta^{-1}w$. 
Then, 
%$U(z')$ is distributed uniformly on the sphere $S_{\C^n}$ while, 
conditionally on $u=U(z)$,
$U(w')$ is distributed uniformly on the sphere $S_{u^\perp}:=S_{\C^n} \cap u^\perp$.
Since
$U(x) = u+ \beta{U}(w')$ and $U(y) = u - \beta{U}(w')$, it follows that the conditional (on $u=U(z)$) distribution of $f(Ux)-f(Uy)$ is the same as that of $f_u : S_{u^\perp} \to \R$ defined by
$$
f_u(v) = f(u+\beta v) - f(u-\beta v) .
$$
As is readily seen, $f_u$ is $2\beta$-Lipschitz and its mean is $0$.
From L\'evy's lemma, applied to $f_u$ and to the $(2n-3)$-dimensional sphere $S_{u^\perp}$, 
we deduce that, conditionally on $u=U(z)$,
\[ \P( |f(Ux)-f(Uy)| > \lambda) \leq C_1 \exp(-c_1(2n-2)\lambda^2/|x-y|^2),\]
and hence the same inequality holds also without the conditioning.
\end{proof}

The end of the proof (the actual chaining argument) is identical to that in  Schechtman's paper, so --- rather than 
copying it --- we present the general principle on which it is based.   Let $(S,\rho)$ be a compact metric space 
and let $\big(X_s\big)_{s\in S}$ be a family of mean $0$ random variables (a stochastic process indexed by $S$). We say that $\big(X_s\big)$ is {\itshape subgaussian} if there are $A,\alpha >0$ such that, for all $s,t \in S$ with $s \neq t$ and for all $\lambda\geq 0$,
\begin{equation} \label{subgaussian-tails} \P( |X_s-X_t| \geq \lambda) \leq A \exp\left(-\alpha \frac{\lambda^2}{\rho(s,t)^2}\right), \end{equation}

\begin{prop}[Dudley's inequality] \label{dudleyineq}
If $(X_s)_{s \in S}$ satisfies \eqref{subgaussian-tails} and some mild regularity conditions,
\vskip-6mm
\[ \E \sup_{s,t \in S} |X_s-X_t| \leq C'A\alpha^{-1/2} \int_0^\iy \sqrt{\log N(S, \eta)}\, d\eta .\] 
where $N(S,\eta)$ is the %metric entropy 
minimal cardinality of a $\eta$-net of $S$ (in particular the integrand is $0$ if $\eta$ is larger than the radius of $S$). 
\end{prop}

See \cite{dudley} for the original article, \cite{jm} for a generalization to the subgaussian 
case that is relevant here, and \cite{talagrand} for a book exposition; we also sketch a proof further below for the reader's convenience.
 
In our case we choose $S = S_{\C^m} \cup \{0\}$ (with the usual 
Euclidean metric), $X_s = f(Us) - \mu$ if $s\in S_{\C^m}$ and $X_0=0$ ; then
\vskip-4mm
\[ \osc(f \circ U, S_{\C^m}, \mu) = \sup_{x \in S} |X_s| .\]
 The underlying probability space is $\mU(n)$, and the subgaussian property is given by Lemma \ref{lemma-subgaussian} if $s, t \in S_{\C^m}$ and directly by L\'evy's lemma if $s$ or $ t$ equals $0$. Next, the bound 
$N\big(S_{\C^m},\eta\big)  = N\big(S^{2m-1},\eta\big) \leq  (1+2/\eta)^{2m}$ mentioned in the comments following Lemma \ref{levy} leads to 
an estimate $2\sqrt{m}$ for the integral and to the bound 
\vskip-2mm
$$
E:= \E \sup_{s\in S} |X_s| \leq \E \sup_{s, t\in S} |X_s-X_t|  \leq C'C(cn)^{-1/2} \cdot 2\sqrt{m} = C''\sqrt{\frac mn} .
$$
(For readers confused by different quantities appearing on the left side in different forms of Dudley's inequality, 
we point out that the first inequality above uses the fact that one of the variables $X_t$ equals $0$, 
and that we always have $\sup_{s, t} |X_s-X_t|  = \sup_{s}  X_s + \sup_{t} (-X_t)$.) 
The assertion of Theorem \ref{dvoretzky-sharp}  follows now from Markov's inequality if $\e$ 
is sufficiently larger than $E$, which is assured by choosing $c_0$ small enough. 
A slightly more careful argument (such as that given in \cite{schechtman}, or see \cite{talagrand}) 
or an application of  the appropriate concentration inequality (for functions on $\mU(n)$) 
yields a bound of the form  $\exp(-c'\e^2n)$ on the probability of the exceptional set 
$\sup_{s\in S} |X_s|  > C''\sqrt{\frac mn}+\e$ (hence for the exceptional set from Theorem \ref{dvoretzky-sharp}). 
%$\sup_{s\in S} |X_s|  > E+\e$. 

Let us comment here that the value of the constant $c_0$ given by the proof of Theorem 
\ref{dvoretzky-sharp} is probably the single most important obstacle to showing 
Theorem  \ref{maintheorem} for ``reasonable'' values of $k, m$.  An adaptation of the proof 
from \cite{gordon} (which yields good constants) to the complex case 
%--- if possible ---  
could be helpful here.

\begin{proof}[Proof of Dudley's inequality]
For every $k \in \Z$, let $\mN_k$ be a $2^{-k}$-net of minimal cardinality for $(S,\rho)$. Let $k_0 \in \Z$ such that the radius of $S$ lies between $2^{-(k_0+1)}$ and $2^{-k_0}$; the net $\mN_{k_0}$ consists of a single element $s_0$. For every $s \in S$ and $k \in \Z$, let    $\pi_k(s)$ be an element
of $\mN_k$ satisfying $\rho(s,\pi_k(s)) \leq 2^{-k}$. The {\itshape chaining equation} reads for every $s \in S$
\begin{equation} \label{chaining1} X_s = X_{s_0} + \sum_{k \geq k_0} X_{\pi_{k+1}(s)} -X_{\pi_k(s)} .\end{equation}
(It is here where some regularity of $(X_s)$ -- path continuity -- is used.) It follows that
\begin{equation} \label{chaining2}
\sup_{s,t \in S} |X_s-X_t| \leq 2 \sum_{k \geq k_0} \sup_{s \in S} | X_{\pi_{k+1}(s)} -X_{\pi_k(s)} | \leq 2 \sum_{k \geq k_0} \sup_{u,u'} |X_u-X_{u'}|, \end{equation}
where the last supremum is taken over couples $(u,u') \in \mN_{k+1} \times \mN_{k}$ satisfying $\rho(u,u') \leq 2^{-k}+2^{-(k+1)}< 2^{-k+1}$.
It remains to bound the expectation of each term in the sum, using the following fact

\begin{fact} \label{max-N-gaussian}
If  $N \geq 2$ and $Y_1,\dots,Y_N$ are nonnegative random variables satisfying the tail estimate $ \P( Y_i \geq t) \leq A \exp (-t^2/2\beta^2)$ for all $t \geq 0$, then
\[ \E \max Y_i \leq C A \beta \sqrt{\log N} .\]
\end{fact}
To bound $\E \sup | X_u-X_{u'}|$, we apply the above fact with $\beta = 2^{-k+1} \alpha^{-1/2}$ and $N= \card(\mN_k) \cdot \card(\mN_{k+1}) \leq N(S,2^{-(k+1)})^2$. This gives
\[ \E \sup_{s,t \in S} |X_s-X_t| \leq C'A\alpha^{-1/2} \sum_{k \geq k_0} 2^{-k} \sqrt{\log N(S,2^{-(k+1)})} \]
The result now follows by relating the last series to the integral in Proposition \ref{dudleyineq} (a version of the integral test from calculus).
\end{proof}

\begin{proof}[Proof of Fact \ref{max-N-gaussian}]
We may assume $\beta = 1$ by working with $Y_i/\beta$. Then simply write
\[
\E \max Y_i = \int_0^\infty \P (\max Y_i \geq t) dt
\leq \sqrt{2\log N} + AN \int_{\sqrt{2\log N}}^\infty \exp(-t^2/2) dt
\leq \sqrt{2\log N} + A
.\]
The last inequality follows from $\displaystyle \int_{\sqrt{2\log N}}^\infty \exp(-t^2/2) dt \leq \int_{\sqrt{2\log N}}^\infty t \exp(-t^2/2) dt = 1/N$. Note that the hypotheses force $A\geq 1$.
\end{proof}

\section*{Appendix B  : Proof of lemma \ref{concentration-spectrum}}

The lemma will follow if we show that with large probability, 
\[ \| \Delta \|_{\iy} \leq \frac{C}{\sqrt{kd}}, \]
where $\Delta=MM^\dagger-\Id/k \in \M_k$ and $\| \cdot\|_{\iy} $ is the operator (or spectral) norm. Let $\mN$ be a $\frac{1}{4}$-net of $S_{\C^k}$ with cardinality bounded by $(C_0)^k$. One checks that if $x \in S_{\C^k}$ and $\bar{x} \in \mN$ satisfy $|x-\bar{x}|\leq 1/4$, then
\[ \left| \bra{x} \Delta \ket{x} \right| \leq \left| \bra{\bar x} \Delta \ket{\bar x} \right| + \left| \bra{x-\bar x} \Delta \ket{\bar x} \right| + \left| \bra{x} \Delta \ket{x-\bar x} \right| \leq \left| \bra{\bar x} \Delta \ket{\bar x} \right| + 2 \cdot
\frac{1}{4} \|\Delta\|_{\iy},\]
so that taking supremum over $x \in S_{\C^k}$, we get
\[ \| \Delta \|_{\iy} \leq 2 \sup_{\bar x \in \mN} \left| \bra{\bar x} \Delta \ket{ \bar x} \right| .\]
An application of the union bound gives
\begin{eqnarray*} \P \left( \| \Delta \|_{\iy} \geq \frac{C}{\sqrt{kd}} \right) &\leq& (C_0)^k \cdot \P \left( \bra{x_0} \Delta \ket{x_0} 
\geq \frac{C}{2\sqrt{kd}}\right) \\
& =&  (C_0)^k \cdot \P \left( |M^\dagger x_0|^2 \geq \frac{1}{k} + \frac{C}{2\sqrt{kd}} \right) \\
& \leq & (C_0)^k \cdot \P \left( |M^\dagger x_0| \geq \frac{1}{\sqrt{k}} + \frac{C}{5\sqrt{d}} \right)
\end{eqnarray*}
where $x_0 \in \C^k$ is any fixed unit vector (remember that $d\geq C^2k$). The probabilities above can be expressed in terms of Beta-type integrals, but it's easier to estimate them using L\'evy's lemma. The function $M \mapsto |M^\dagger x_0|$ is $1$-Lipschitz on the Hilbert--Schmidt sphere (if $x_0$ is the first vector of the canonical basis, then $ M^\dagger x_0$ is essentially the first row of $M$) and
\[ \E |M^\dagger x_0| \leq \left(\E |M^\dagger x_0|^2\right)^{1/2} = \sqrt{1/k} .\]
Hence, by L\'evy's lemma (with $n=2kd$ and $\e = \frac{C}{5\sqrt{d}}$), we get
\[ \P \left( \| \Delta \|_{\iy} \geq \frac{C}{\sqrt{kd}} \right) \leq \exp(-c k)\]
for some choice of the constants $C,c>0$, as required. 
\vskip1cm


\begin{thebibliography}{10}

\nocite{}

\bibitem{nc}
Nielsen, M.~A.,  Chuang, I.~L.: {\itshape Quantum computation and quantum information}.
Cambridge University Press, Cambridge (2000)

\bibitem{holevo-icm} 
Holevo, A.~S.: {\itshape The additivity problem in quantum information theory}.  
In ``Proceedings of the International Congress of Mathematicians (Madrid, 2006),'' 
Vol. III,  999--1018, Eur. Math. Soc., Z\"urich (2006)

\bibitem{shor} Shor, P. W.:
{\em Equivalence of additivity questions in quantum information theory}.  
Comm. Math. Phys.  {\bfseries 246},  no. 3, 453--472   (2004)

\bibitem{hastings}
Hastings, M. B.:
{\em Superadditivity of communication capacity using entangled inputs.} 
Nature Physics {\bf 5}, 255 (2009)

\bibitem{dvoretzky} 
Dvoretzky,  A.:
{\em Some Results on Convex Bodies and Banach Spaces.} 
In: {\em Proc. Internat. Sympos. Linear Spaces (Jerusalem, 1960)}, pp. 123--160. 
Jerusalem Academic Press, Jerusalem; Pergamon, Oxford (1961)

\bibitem{bh}
Brandao, F.,  Horodecki, M.: 
{\itshape On Hastings' counterexamples to the minimum output entropy additivity conjecture}. 
%\verb!arXiv:0907.3210v1 [quant-ph]!
Open Syst. Inf. Dyn. {\bfseries 17}, 31 (2010), e-print arxiv:0907.3210v1 [quant-ph]

\bibitem{FKM}   Fukuda, M., King, C., Moser, D.: 
  {\em  Comments on Hastings' Additivity Counterexamples.} 
% \verb!arXiv:0905.3697!
Commun. Math. Phys. 296, 111 (2010); e-print arxiv:0905.3697 [quant-ph].

\bibitem{hw} Hayden, P.,  Winter, A.:                                                                               
{\em Counterexamples to the maximal $p$-norm
multiplicativity conjecture for all $p > 1$.}
Comm. Math. Phys. {\bfseries 284}, 263--280, (2008);  %\verb!arXiv:0807.4753! 
e-print arxiv:0807.4753v1 [quant-ph]

\bibitem{asw}
Aubrun, G., Szarek, S.,  Werner, E.: 
{\itshape Non-additivity of R\'enyi entropy and Dvoretzky's theorem.} 
J. Math. Phys. {\bfseries 51}, 022102 (2010)

\bibitem{milman} 
Milman, V.:
{\em A new proof of the theorem
of A. Dvoretzky on sections of convex bodies.} 
Funct. Anal. Appl.  {\bfseries 5} (1971), 28--37 (English translation)

\bibitem{FLM}   
Figiel, T., Lindenstrauss, J.,  Milman, V. D.: 
 {\em  The dimension of almost spherical sections of convex bodies.}
Acta Math.  {\bfseries 139}, no. 1-2, 53--94  (1977)

\bibitem{cn}
Collins, B., Nechita,   I.: 
{\em Gaussianization and eigenvalue statistics for random quantum channels (III)}, 
Ann. Appl. Probab., to appear; e-[print arxiv:0910.1768v2 [quant-ph]

\bibitem{levy} L\'evy, P.:
{\itshape Probl\'emes concrets d'analyse fonctionnelle},  2nd ed. 
Gauthier-Villars, Paris (1951)

\bibitem{pisier} Pisier, G.:
{\em The volume of convex bodies and Banach space geometry.}
Cambridge Tracts in Mathematics,  94.
Cambridge University Press, Cambridge (1989) 

\bibitem{gordon}
Gordon, Y.: 
{\itshape On Milman's inequality and random subspaces which escape through a mesh in $\R^n$}.  
In:  ``Geometric aspects of functional analysis (1986/87),''  Lecture Notes in Math., 1317, pp. 84--106. 
Springer, Berlin (1988)

\bibitem{schechtman}
Schechtman, G.: 
{\itshape A remark concerning the dependence on $\e$ in Dvoretzky's theorem}.  
In: ``Geometric aspects of functional analysis (1987--88),''  
Lecture Notes in Math., 1376,  pp. 274--277, Springer, Berlin (1989)

\bibitem{MP}  Marchenko, V. A.,  Pastur, L. A.:
{\em The distribution of eigenvalues in certain sets of random matrices.} 
Mat. Sb.  {\bfseries 72}, 507-536 (1967)

\bibitem{hlw} Hayden, P., Leung, D.,  Winter, A.: 
{\em Aspects of generic entanglement.} Comm. Math. Phys. {\bfseries 265}, 95--117 (2006)
 
\bibitem{HT}
Haagerup,  U.,  Thorbj\o rnsen, S.: 
 {\em Random matrices with complex Gaussian entries.}
  Expositiones Math.  {\bfseries 21}, 293--337  (2003) 

\bibitem{geman} Geman, S.:
{\it A limit theorem for the norm of random matrices}. 
Ann.\ Probab.\  {\bfseries 8}, 252--261  (1980)

\bibitem{silver} Silverstein, J. W.: 
{\it The smallest eigenvalue of a large-dimensional Wishart matrix},
Ann.\ Probab.\ {\bfseries  13}, 1364--1368,  (1985)

\bibitem{dudley}  Dudley, R. M.: 
{\em The sizes of compact subsets of Hilbert space and continuity of Gaussian processes.} 
 J. Funct. Anal. {\bfseries 1}, 290--330,  (1967)

\bibitem{jm} 
Jain, N. C.,  Marcus, M. B.:
{\em Continuity of subgaussian processes. } 
In: ``Probability on Banach Spaces,''  Advances in Probability, Vol. 4, 81--196, Dekker, New York 1978.

\bibitem{talagrand} Talagrand, M.: 
 {\itshape The generic chaining. Upper and Lower bounds of Stochastic Processes.} 
 Springer 2005.
 


\end{thebibliography}
\end{document}